\documentclass[a4paper]{article}
\usepackage[latin1]{inputenc} 
\usepackage[T1]{fontenc} 
\usepackage{RR}
\usepackage{hyperref}
\usepackage{subfigure}

\newtheorem{theorem}{Theorem}[section]

\newenvironment{proof}[1][Proof]{\begin{trivlist}
\item[\hskip \labelsep {\bfseries #1}]}{\end{trivlist}}
\RRdate{November 2011}

\RRauthor{
Alexandre Mouradian
  \thanks[sfn]{University of Lyon, INRIA, France, INSA Lyon, CITI, F-69621, France}%
  \and
Isabelle Aug\'e-Blum
\thanksref{sfn}
}
\authorhead{Mouradian \& Aug\'e-Blum}
\RRtitle{1-D Coordinate Based on Local Information for MAC and Routing Issues in WSNs}
\RRetitle{1-D Coordinate Based on Local Information for MAC and Routing Issues in WSNs}
\titlehead{1-D Coordinate}
\RRresume{De plus en plus d'applications des r\'eseaux de capteurs sans fil \'emergent avec de fortes contraintes applicatives notamment en terme de fiabilit\'e et de respect de contraintes temporelles. Les m\'ecanismes r\'eseaux sous-jacents tel que les protocoles MAC et de routage doivent pouvoir satisfaire ces contraintes. Dans le cas des contraintes temporelles notre approche du probl\`eme consiste \`a borner le nombre de sauts maximum entre les noeuds et le puits, ainsi que la dur\'ee d'un saut, cela permettant de borner le d\'elai de bout en bout. Concernant la fiabilit\'e, nous proposons de s\'electionner les noeuds relais en fonction de leur connectivit\'e avec les noeuds plus proches du puits. Pour pouvoir parvenir \`a cela nous avons besoin d'une coordonn\'ee (ou une m\'etrique) qui donne une information sur le nombre de sauts, qui permet de fortement diff\'erencier les noeuds et qui donne une information sur la connectivit\'e de chaque noeud. Le routage de type ``greedy`` est consid\'er\'e comme efficace dans les r\'eseaux de capteurs sans fil car il permet le passage \`a l'\'echelle. Cependant le co\^ut financier d'un module GPS a pouss\'e \`a d\'evelopper des syst\`emes de coordonn\'ees virtuelles. Un type de coordonn\'ees est bas\'e sur le nombre de saut entre les noeuds et le puits, cela implique que plusieurs noeuds peuvent avoir la m\^eme coordonn\'ee. Le principal avantage de ce syst\`eme est que les noeuds connaissent le nombre de sauts qui les s\'eparent du puits. Cependant, il ne permet pas de diff\'erencier les noeuds qui ont le m\^eme nombre de sauts. Dans ce rapport nous proposons un nouveau syst\`eme de coordonn\'ees qui a pour but de classer les noeuds qui ont le m\^eme nombres de sauts en fonction leur connectivit\'e et de diff\'erentier les noeuds dans un 2-voisinage. Ces propri\'et\'es font de cette coordonn\'ee, qui peut \^etre vue comme un identifiant local, une m\'etrique puissante qui peut \^etre utilis\'ee dans les m\'ecanismes des r\'eseaux de capteurs sans fil.}

\RRabstract{More and more critical Wireless Sensor Networks (WSNs) applications are emerging. Those applications need reliability and respect of time constraints. The underlying mechanisms such as MAC and routing must handle such requirements. Our approach to the time constraint problem is to bound the hop-count between a node and the sink and the time it takes to do a hop so the end-to-end delay can be bounded and the communications are thus real-time. For reliability purpose we propose to select forwarder nodes depending on how they are connected in the direction of the sink. In order to be able to do so we need a coordinate (or a metric) that gives information on hop-count, that allows to strongly differentiate nodes and gives information on the connectivity of each node keeping in mind the intrinsic constraints of WSWs such as energy consumption, autonomy, etc. Due to the efficiency and scalability of greedy routing in WSNs and the financial cost of GPS chips, Virtual Coordinate Systems (VCSs) for WSNs have been proposed. A category of VCSs is based on the hop-count from the sink, this scheme leads to many nodes having the same coordinate. The main advantage of this system is that the hops number of a packet from a source to the sink is known. Nevertheless, it does not allow to differentiate the nodes with the same hop-count. In this report we propose a novel hop-count-based VCS which aims at classifying the nodes having the same hop-count depending on their connectivity and at differentiating nodes in a 2-hop neighborhood. Those properties make the coordinates, which also can be viewed as a local identifier, a very powerful metric which can be used in WSNs mechanisms.}
\RRmotcle{r\'eseaux de capteurs sans fil, syst\`eme de coordonn\'ees virtuelles, contraintes applicatives, MAC, routage}
\RRkeyword{wireless sensor networks, virtual coordinate system, application constraints, MAC, routing}
 \RRprojet{Swing}  
\RRprojet{swing} 
 \RCGrenoble 

\begin{document}
\RRNo{7819}
\makeRR   

\section{Introduction}

In this document we focus on low convergecast traffic WSNs applications where alarms from any particular nodes must reach the sink in a bounded time with a given reliability. We can cite for example volcano monitoring~\cite{Tan10} and forest fires detection~\cite{Zhang08} applications.

The WSNs mechanisms such as MAC, routing and data aggregation (before the alarm is forwarded toward the sink) need to have capabilities to handle such critical applications. Our approach to the time constraint problem is to bound the hop-count between a node and the sink and the time it takes to do a hop. So the end-to-end delay can be bounded. The bound on the hop duration implies a MAC mechanism which avoids packet collisions. Thus the nodes have to be strongly differentiated to be able to make a decision on which node accesses the medium at a given time. At routing layer the length of any path between a source and the sink has to be known and bounded in order to give guaranties on end-to-end delay. In convergecast networks the hop-count-based solutions such as \cite{Ye05} allow this. Nevertheless it does not allow to differentiate the nodes for forwarder selection because many nodes have the same hop-count. For reliability purpose, the forwarder selection should be based on node's connectivity and the nodes having more neighbors in proportion with smaller hop-count should be preferred. In data aggregation context, a node that gather the data is needed. The choice of this node must be deterministic to avoid too long delays. In this report, we focus on MAC and routing protocols.

In this document we propose a 1-D coordinate. The key ideas of our proposition are to classify the nodes having the same hop-count and strongly differentiate them in a 2-hop neighborhood. We do not need to differentiate the nodes in the whole network because MAC and routing mechanisms are usually localized at a 2-hop neighborhood level. At MAC level it is due to the hidden terminal problem and at routing layer a node must choose a forwarder between its neighbors. Our proposition uses only local information in order to build the coordinate thus it is scalable.

In Section 2, an overview of the advantages and drawbacks of existing solutions for WSNs is presented. In Section 3, hypothesis and requirements are given. In Section 4, the theoretical reflexion followed to construct the coordinate is explained and possible issues it can induce are discussed. A theoretical analysis of the coordinate is done in section 5. In section 6, a practical solution to construct the coordinates is given. In Section 7, simulations results are presented and the performances of the algorithm and coordinates are discussed. Section 8 concludes on the presented work and lists future works.

\section{Related Work}
In the last years, many VCSs have been proposed. This can be explained by the fact that greedy routing has been proven to be very efficient to perform routing in WSNs mainly because of its stateless characteristic. The first propositions~\cite{Karp00} \cite{Bose99} were based on geographic coordinates. The issue of this solution is the high financial cost of a GPS chip and the number of nodes in WSNs which can reach thousands of units. Moreover, the lack of accuracy in the position of the nodes can induce bad performances of greedy routing \cite{Watteyne07}. These problems led to solutions based on virtual coordinates because the exact location of all the nodes is not necessary. A VCS can be Cartesian~\cite{Rao03}, polar~\cite{Newsome03} or based on anchors~\cite{Caruso05}~\cite{Cao04}. In the first case, the virtual coordinates of nodes are given in the same space as the real ones. In the last case, the coordinates are given in distance from anchor nodes (thus if there are $n$ anchors the node is placed in a n-dimension space). A special case of the last type with only the hop-count from the sink is used in convergecast networks \cite{Ye05}.

A solution based on a Cartesian system is proposed in \cite{Rao03}. First,  perimeter nodes are identified and being given coordinates. Then, each node iteratively updates its coordinates with the center of gravity of its neighbors' coordinates. Nodes others than perimeter ones are initially placed at the center of the area and move toward the borders of the network. \cite{Watteyne09} improves this scheme by constructing the coordinates during the runtime. Moreover it does not need to detect perimeter nodes and it considers the sink as the center of the coordinates system. With those system it is difficult to know the routing path length and connectivity information can not be deduced from the node's coordinate. 

Anchors-based VCSs were proposed in \cite{Caruso05} and \cite{Cao04}. Anchors nodes broadcast messages which contain a counter incremented at each hop. For example, in a case where there are three anchors, by listening to these messages a node can determine its virtual coordinates $(V1,V2,V3)$ where $V1$ (resp. $V2$ and $V3$) is the minimum number of hops from anchor 1 (resp. 2 and 3) to the given node. As we are interested in convergecast networks, we focus more on 1-D anchor systems with the sink being the anchor.

A VPCS (Virtual Polar Coordinates System) is proposed in \cite{Newsome03}. Each node has a coordinate corresponding to its number of hops from the sink and a coordinate corresponding to an angle range. A tree representing the network connectivity is built with the sink node as its root. The sink attributes an angle range to each of its sons depending of the number of nodes it has under it. Each node divides its range between its sons. This scheme has the advantages to give the information about hop count and to differentiate the nodes with the angle parameter. Nevertheless, this last parameter is not physically meaningful because two contiguous angles may be attributed to two different nodes which are not neighbors. The solution is centralized thus not scalable. Moreover, a change in the topology induces a reconstruction of a part of the tree which can be costly in energy.

In \cite{Ye05}, the authors propose GRAB which introduces a cost-field. This cost-field can be seen as a VCS, it represents the cost for a node to reach the sink. In the paper the hop-count is used as a cost-field. Each node is assigned its distance to the sink, in number of hops, as coordinate. Then, packets are routed using gradient-routing which consists in choosing the link with the highest gradient, the gradient being defined by the difference between the cost-fields of two nodes. As many nodes with the same hop count can hear the packet, the selection of the forwarder can be based on a random value and multiple forwarders can be elected, creating multiple paths. The advantages of such a solution are that the number of hops to reach the sink is known and multiple path leads to more reliability. Nevertheless GRAB does not give information on the physical organization of nodes having the same hop-count. SGF \cite{Huang09} and LQER \cite{Chen08} propose similar schemes. In SGF only one node is chosen. LQER adds information on the link quality. Both solutions suffer from the same drawbacks of GRAB.

Of the VCSs proposed in the literature, none can give information on the cost in hop numbers from any given node to the sink and strongly differentiate the nodes in a 2-hop neighborhood at the same time with the differentiation depending on the connectivity of the node. For these reasons we propose a new VCS which provides those properties. It facilitates the deployment of many new mechanisms for WSNs.


\section{Problem statement}
In order to be as general as possible we assume that the sensor nodes have a finite amount of energy. The radio is half duplex and mono-channel and the nodes have no information on their geographic position. In this  context the aim of our solution is to provide a 1-D coordinate that should give information about the physical position in term of hop count of a node from the sink and classify the nodes having the same hop-counts. The coordinate should also allow to strongly differentiate nodes in a 2-hop neighborhood. Our solution should be scalable and energy-aware in order to be deployable in WSNs.

\section{Theoretical data calculation}

We present the theoretical background of our VCS. Our system is composed of only one coordinate which is calculated in function of the number of hops to the sink and an offset which is computed from theoretical data. In a hop-count-based VCS, nodes having the same hop-count form rings centered on the sink (in the theoretical calculations, we suppose perfect radio links and that the nodes repartition is homogeneous). The aim of this coordinate is to give information on the hops number and to classify the nodes within a given ring and in a 2-hop neighborhood. The key idea is to have a coordinate which strongly differentiates nodes in a 2-hop neighborhood and which has a physical meaning in the ring: nodes with proportionally more neighbors in the lower ring are classified before ones having proportionally less neighbors in the lower ring (the lower rings being the nearest from the sink). Classification is done in function of the connectivity of the nodes with the different rings. 

The construction of the coordinates is done in two steps, the first is the theoretical data generation from a theoretical model. The second step is mapping the theoretical data on the network in order to give coordinates to the nodes. We detail those two steps in the next paragraphs.

\subsection{Theoretical model}

Our reflexion is based on the Unit Disk Graph (UDG) model. On Figure~\ref{model}
 we see that the coordinate is constructed by using the information on the number of hops from the sink (noted $n$) and an offset in a ring, $R$ being the radio range. The formula used to compute the coordinate of point $p$ is:
 \[
coord_{p}=(n-1)*R+offset \mbox{ with } offset<R
\]

The offset is used to classify the nodes within a same ring. We assume that the node knows its ring (number of hops from the sink) and the number of neighbors it has at each ring (or at least at $n-1$, $n$ and $n+1$; $n$ being the ring of the considered node). The algorithm used to obtain this information is described further. The node then computes the percentage of neighbors it has at each ring and use this information to compute the offset. The idea is to find a mapping between these percentages of neighbors at each ring and the offset of the node in the ring. This is achieved by producing theoretical data where the percentages of neighbors are replaced by percentages of areas of the theoretical range of the considered node in each theoretical ring as shown in Figure~\ref{model} 
(percentages of areas $A$, $B$ and $C$ corresponding respectively to percentages of neighbors in ring $n-1$, $n$ and $n+1$). We insist on the the fact that those areas are theoretical because in reality the range of a node may not be a perfect disk and the rings may not be perfect (if there is a hole in the network for instance). Nevertheless, this theoretical data can actually be used to compute an offset.

\begin{figure}[!h]
  \centering
  \includegraphics[width=2.3in, keepaspectratio=true]{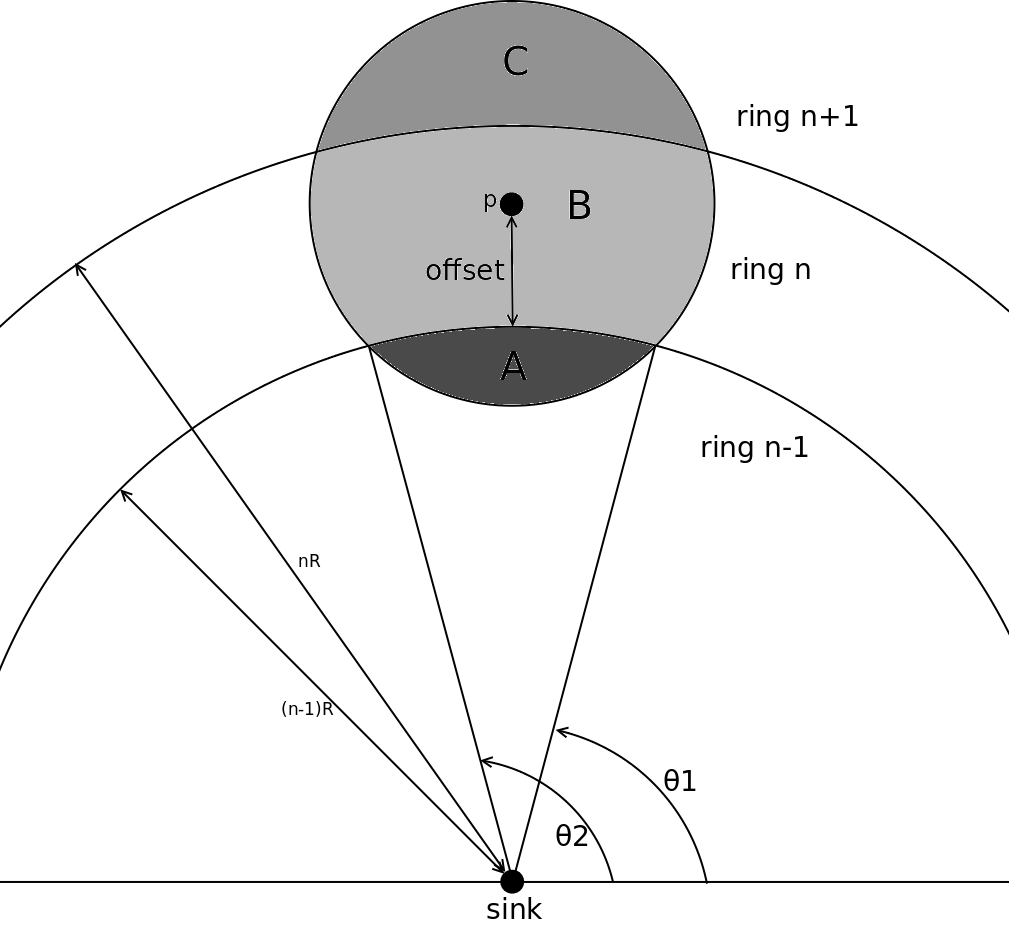}
  \caption{Theoretical model.}
  \label{model}
\end{figure}

We can see on the Figure~\ref{model} that the offset parameter is directly linked to the values of $A$, $B$ and $C$ areas so we can find functions of the type $f(offset)=A$, $g(offset)=B$ and $h(offset)=C$. This is done by calculating one area in function of the position of the node in the ring n. For example, the area A is given by the following integral : \\ \\
\[
A=\int_{\theta_{1}}^{\theta_{2}}\int_{r(\theta)}^{(n-1)R}rdrd\theta
\]

with
\begin{eqnarray*}
r(\theta)&=&[(n-1)R+offset]sin{\theta}\\
& & -\sqrt{R^2-[(n-1)R+offset]^2cos{\theta}^2}
\end{eqnarray*}

let $A_{1}$ be
\[A_{1}=\int_{r(\theta)}^{(n-1)R}rdr=\left[ \frac{r^2}{2} \right]_{r(\theta)}^{(n-1)R}\] 

so
\[ A=\int_{\theta_{1}}^{\theta_{2}}A_{1}d\theta \]
\begin{eqnarray*}
  A & = &\frac{(n^2-2n)R^2}{2} \left[ \theta \right]_{\theta_{1}}^{\theta_{2}} + \frac{\left[ (n-1)R+offset \right]^2}{2} \\
  & & \times \left[ \sin \theta \cos \theta \right]_{\theta_{1}}^{\theta_{2}} + \frac{\left[ (n-1)R+offset \right]}{2} \\
  & & \times \left[ -\cos \theta \sqrt{R^2-[(n-1)R+offset]^2 \cos^2 \theta} \right.\\ 
  & & \left. - \frac{R^2 \tan^{-1}(\frac{[(n-1)R+offset]\cos \theta}{\sqrt{R^2-[(n-1)R+offset]^2 \cos^2 \theta}})}{(n-1)R+offset} \right]_{\theta_{1}}^{\theta_{2}} \\
\end{eqnarray*}

\[ \mbox{with}\left\{
  \begin{array}{ll}
   \theta_{1}=\arctan(\frac{(n-1)^2R^2-R^2+[((n-1)R+offset)]^2}{2[((n-1)R+offset)]} \\ \\ \times \frac{+1}{\sqrt{(n-1)^2R^2-[\frac{(n-1)^2R^2-R^2+[((n-1)R+offset)]^2}{2[((n-1)R+offset)]}]^2}})\\ \\
   \theta_{2}=\arctan(\frac{(n-1)^2R^2-R^2+[((n-1)R+offset)]^2}{2[((n-1)R+offset)]} \\ \\ \times \frac{-1}{\sqrt{(n-1)^2R^2-[\frac{(n-1)^2R^2-R^2+[((n-1)R+offset)]^2}{2[((n-1)R+offset)]}]^2}})+\pi\\
  \end{array} 
\right.
\]

On the same principle we can compute $C$ and $B$ is given by $\pi R^{2}-A-C$. We see that for a given offset we obtain values of $A$, $B$ and $C$ so by dividing theses areas by the area of the theoretical range we deduce the offset in function of the percentages of areas $A$, $B$ and $C$. The principle is then to map the percentages of neighbors on the percentages of areas and thus being able to give an offset to each node.

We can notice that two nodes in the same 2-hop neighborhood having the same percentages of neighbors at each ring are given the same coordinate. From now we refer to this situation as a coordinate collision. We see in the evaluation section in that even if this situation can occur it is actually very rare. 

\subsection{Mapping issues}

At this point we have a function that links the percentages of areas with the offset. Now the aim is to give a coordinate to a node which knows the percentages of neighbors it has at rings $n-1$, $n$ and $n+1$ (noted $\%(n-1)$, $\%n$ and $\%(n+1)$). So we have to link those percentages with the area percentages. This is done by a projection of the neighbors percentages values on the areas percentages. This projection is done by choosing in the areas percentages points the one which has the minimum Euclidean distance with the neighbors percentages.

In reality a node can have percentages of neighbors that not fit the theoretical values, for instance if a node does not have any neighbors in its own ring (Figure \ref{model} shows that area B is never null with $0\leq offset < R$). This implies that the space of real percentages values is larger than the theoretical one. Figure \ref{cube} represents the plane of neighbors percentages (the percentages are in the plane $\%(n-1)+\%n+\%(n+1)=1$) and the curve that links $\%A$, $\%B$ and $\%C$ which is also on the plane (because $\%A+\%B+\%C=1$). The projection of the values of the plane on the curve leads to nodes having different neighbors percentages being given the same areas percentages and thus the same offset as pictured in Figure \ref{cube} with points $p$ and $q$. This issue is mitigated by adding to the offset the euclidean distance of the projection.

\begin{figure}[!h]
  \centering
  \includegraphics[width=2.5in, keepaspectratio=true]{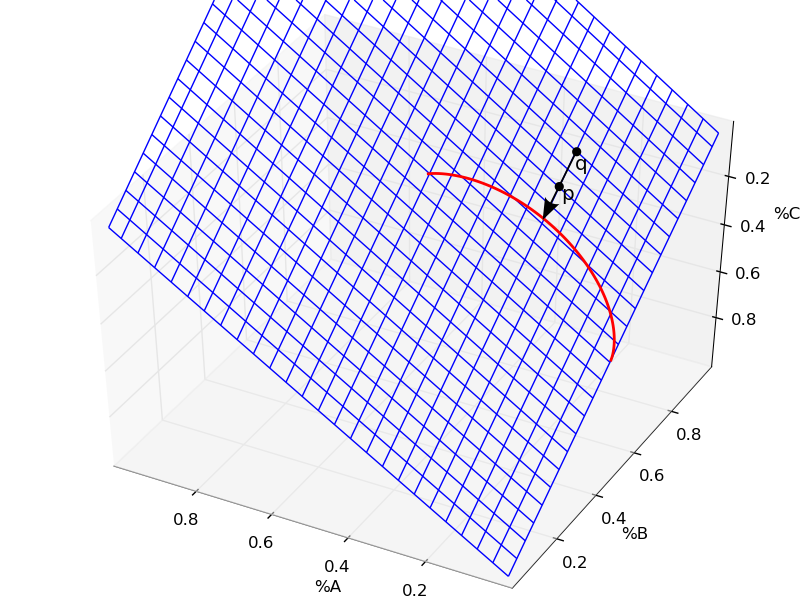}
  \caption{Curve that links \%A, \%B and \%C values}
  \label{cube}
\end{figure}

\begin{theorem}
The addition of the projection distance resolves collisions without adding more if the offset values space is discrete and theoretical consecutive offset values are separated by at least the maximum projection distance (the separation is noted $\Delta offset$). 
\end{theorem}
\begin{proof}
 We do a proof by contradiction. Lets suppose the addition of the projection distance creates a collision. It means two points ($p_{1}$ and $p_{2}$) which do not get the same offset ($offset_{1}$ and $offset_{2}$ with $offset_{1}<offset_{2}$) but with the addition the values end to be the same. It is possible if the distance $d_{1}$ associated with $p_{1}$ is $d_{1}=offset_{2}+d_{2}-offset_{1}$. We have $offset_{2}-offset_{1} \geq \Delta offset$, $\Delta offset$ being the distance separating two consecutive offset values, so $offset_{2}-offset_{1}+d_{2}>\Delta offset$ so $d_{1}>\Delta offset$ which is a contradiction.
\end{proof}



In practice a node will embed a table that contains the points of the curve $f(\%A,\%B,\%C)=offset$. The point calculated with the percentages of neighbors will be projected on the nearest point in the theoretical data table, the corresponding offset will be given to the node and the distance between the points will then be added to the offset. This technique allow to have a relatively low granularity of the theoretical data because the addition of the distance will prevent collisions. This property is interesting because the nodes of WSNs have generally a low memory.

In theory, collisions can occur either because nodes in the same 2-hop neighborhood have the same percentage of neighbors in each ring or because the projection distance is the same. But in practice the calculations of the percentages and the distance are done with a finite precision, that induces more collision. We analyze those issues in the next sections.

\section{Theoretical analysis of the solution}
In this section we analyze the theoretical probability of coordinate collision. To do so, we characterize the coordinate space in the neighborhood of a node and compute the expected number of pairs of nodes which have the same coordinate seen by this node. We assume that the nodes are distributed randomly on a plane, that a node in ring $n$ always has at least a neighbor in ring $n-1$ ($\%(n-1)>0$) and that the coordinate is chosen randomly with a uniform law. This last statement implies that we assume that the positions of the neighbors of nodes being in the neighborhood of the same node are independent and thus that there are more probabilities that nodes can differentiate themselves. We do not take into account collisions due to the projection in this part it means that only nodes having the same proportions of neighbors are in a coordinate collision. These hypothesis implies that we are taking into account less collisions in the theoretical analysis than the ones that can occur in reality. Nevertheless this analysis is useful to evaluate the quality of our proposition.


\subsection{Coordinate space characterization}
Here we characterize the coordinate space. We show how the nodes can differentiate them with their proportions of neighbors in the different rings. We consider a node with $k$ neighbors which all have $k$ neighbors. Each node will have a proportion of neighbors in rings $n-1$, $n$ and $n+1$. The combinations of proportions in each ring represent the coordinate space. The number of accessible proportion depends on the number of neighbors. If a node has 2 neighbors it can have 0\% or 50\% of them in ring $n$ (100\% is impossible because it always has at least one node in ring $n-1$), if it has 3 neighbors it can have 0\%, 33\% or 66\% in ring $n$. Thus the cardinal of coordinate space (noted $N$) increase with the number of neighbors. We note that the possibilities of proportion in a ring depend on the others. For example, if a node has 3 neighbors and it has 33\% of them in ring $n-1$ it can have 33\% in ring $n$ and 33\% in $n+1$ or 66\% in $n$ and 0\% in $n+1$ or 0\% in $n$ and 66\% in $n+1$ leaving no other possibilities. 
So we have $\%(n-1)+\%n+\%(n+1)=1$ which can be written 
\[
m\frac{1}{k}+o\frac{1}{k}+p\frac{1}{k}=1
\]
\[
 m+o+p=k
\]
with k the number of nodes in a neighborhood and $m \in [1,k]$ and $o$, $p \in [0,k]$ respectively the numbers of neighbors at $n-1$, $n$ and $n+1$. There are k possibilities for value m. If m is fixed we have $k-m=o+p$ so 
\[
 \left.
  \begin{array}{rcr}
   o+p & = & (k-m)+0 \\
   \mbox{or } o+p & = & (k-m-1)+1 \\
   \mbox{...} \\
   \mbox{or } o+p & = & 0+(k-m) \\
  \end{array}
\right\} k-m+1 \mbox{ possibilities}
\]
We sum the possibilities for each $m$:
\[
 \sum_{m=1}^{k} k-m+1 = k+(k-1)+...+1=\frac{k(k+1)}{2}
\]
Thus the cardinal of coordinate space is 
\[
N=\frac{k(k+1)}{2}
\]
This argument holds if we fix $o$ or $p$ first.

\subsection{Expected number of coordinate collisions}
In this section we compute the expected number of collisions in function of the number of neighbors of the nodes in the network (noted $k$). 
The probability that a given node $i$ has the same coordinate of a node $j$ is $X_{ij}=\frac{1}{N}$ so $X=\sum_{i\ne j} X_{ij}$ is a random variable that represent the number of 2-collisions seen by a node which has $k$ neighbors. We have $E[X]=\sum_{i\ne j} E[X_{ij}]$ with $E[X]$ being the expected number of collisions.
\[
 E[X]=\left(\!\!\!
  \begin{array}{c}
    k \\
    2
  \end{array}
  \!\!\!\right) \frac{1}{N}
\]
We have $N=k(k+1)/2$ so

\begin{eqnarray*}
  E[X] & = & \frac{k!}{2!(k-2)!} \frac{2}{k(k+1)}\\
  & = &  \frac{k!(k-1)k}{2\quad k!} \frac{2}{k(k+1)}\\
  & = & \frac{k-1}{k+1}
\end{eqnarray*}

and \[
     \lim_{k \to +\infty} \frac{k-1}{k+1} = 1
    \]

\begin{figure}[!h]
  \centering
  \includegraphics[width=2.5in, keepaspectratio=true]{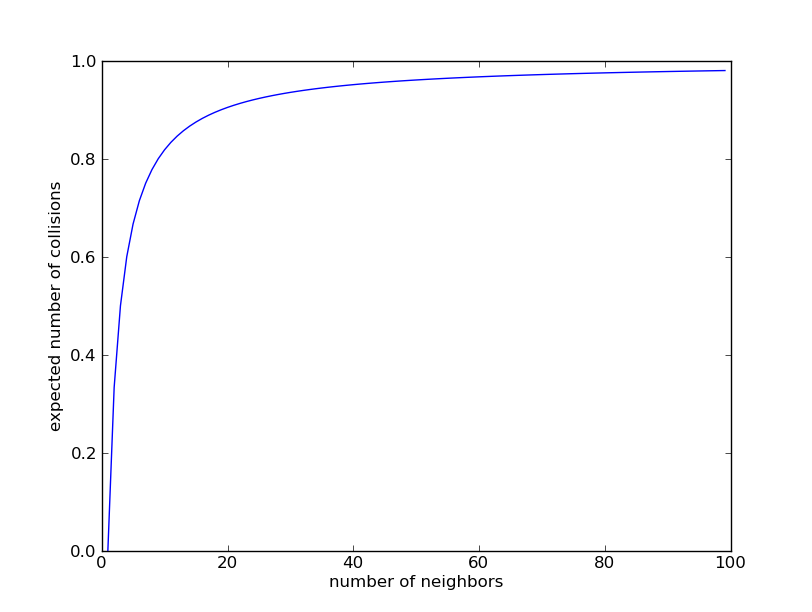}
  \caption{Expected number of collisions in function of the size of a neighborhood}
  \label{theo}
\end{figure}

Figure \ref{theo} represents the plot of $E[X]$ in function of $k$. The curve is always under the value 1 which means that the expected 2-collisions number is bounded by 1. The expected number of coordinate collisions in a 2-hop neighborhood does not depend on the average degree of the network. Nevertheless, in reality there are more collision on average as described in section 7. This is due to the collisions induced by the projection as mentioned in previous section, the fact that the repartitions of the neighbors of nodes that are neighbors are not independent and also because of the use of finite precision number in the implementation which is described in section 7. The result is still interesting because it shows that the number of collision should be stable whatever the density of the network is.

\section{Practical construction of the coordinate}
In the previous sections we saw the key principles and theoretical analysis of the VCS. In this section we will focus on how the nodes can gather information about their hop count and the percentages of neighbors they actually have in the different rings.

We can notice that the nodes are using a duty-cycle \cite{Polastre04} mechanism. They alternatively wake up and go into sleep state. This mechanism reduces the amount of energy consumed during the initialization of the coordinates.

During the initialization a node gets information about in which ring it is and the number of neighbors it has in the different rings. There are two versions of the algorithm, one synchronous were the nodes know when their neighbors wake up and another asynchronous in which they have no information on wakeup dates. Here we will describe only the asynchronous algorithm since the synchronous is the same without the part which synchronize the nodes (because they are assumed to be synchronized by another mechanism).

The sink begins the algorithm, it starts an initialization that reaches all the nodes. The algorithm is described for a node at ring $n$. The nodes are synchronized with a long preamble \cite{Polastre04}. Nodes at $n-1$ ring send a radio noise used to synchronize nodes at rings $n-1$ and $n$. Then there is a slotted contention period where the nodes at ring $n-1$ chose randomly a slot (using a uniform law). They send a packet containing their ring number. It allows the nodes at ring n to know in which ring they are (the ring number they receive in the messages plus one) and by counting the number of packet received they deduce the number of neighbors they have at ring $n-1$. The process is repeated with nodes at layer $n$ and $n+1$ thus at the end of three contention periods a node knows its ring number and the number of neighbors it has at ring $n-1$, $n$ and $n+1$. 

\begin{figure}[!h]
  \centering
  \includegraphics[width=4in, keepaspectratio=true]{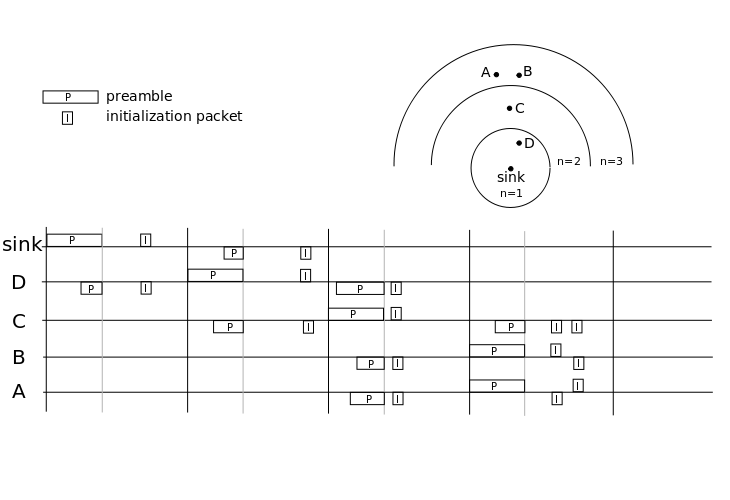}
  \caption{Example of network initialization}
  \label{example}
\end{figure}

Figure~\ref{example} depicts the initialization algorithm with 4 nodes and the sink. The sink starts the process by sending the first preamble. $D$ wakes up and listens to the end of the preamble so it synchronizes with the sink. The first initialization packet informs $D$ that it is in ring 1. $D$ emits a preamble at the end of the contention period, this preamble synchronizes the sink, $C$ and $D$. It then sends an initialization packet in the contention phase that follows so the sink node and $C$ knows that they have one neighbor in ring 1. At the end of the contention period $C$ sends a preamble that synchronizes $A$, $B$, $C$ and $D$. $A$, $B$ and $D$ then receive the initialization packet of $C$. The same process is repeated with $A$ and $B$ synchronizing with $C$ and sending initialization packets. Thus at the end of the process every node has received one initialization packet from each of its neighbors and so it knows the number of neighbors it has in each ring.

The nodes have to listen to three contentions periods in which they receive packets from their neighbors at the different layers. They send a packet only once. Thus the energy consumed during the initialization is the energy needed to listen during three contention periods and to send one packet plus the energy used for the synchronization. The use of a global synchronization or of a long preamble (synchronous or asynchronous version) depends on the application. If global synchronization is needed by the mechanisms which uses the coordinate, it could also be used for the construction of the coordinate.

\section{Performances of the coordinates}
    

We used the discrete events simulator for WSNs, WSNet~\cite{wsnet}. We simulated a network of dimensions 50x50 square units with the sink at (25,25) the communication range is 10 (we chose a relatively small simulation area because it limits the simulations duration : we can have a high increase of the network average degree with a relatively low increase of the number of nodes). We simulated with two different propagation models, the free space propagation model which corresponds only to the path loss without interference or noise, this allow to test our algorithm with a perfect channel. The second is the log-normal shadowing model which has been proven \cite{Zuniga04} to be very suited to model real wireless links in the case of WSNs. We simulate the initialization protocol previously described with 50 to 750 nodes placed randomly in order to increase the density of the network, it represents 3 to 5 hops depending on the topologies. We simulated the asynchronous version of the protocol.

\begin{figure}[!h]

    \subfigure[Free space model]{\includegraphics[width=2.5in, keepaspectratio=true]{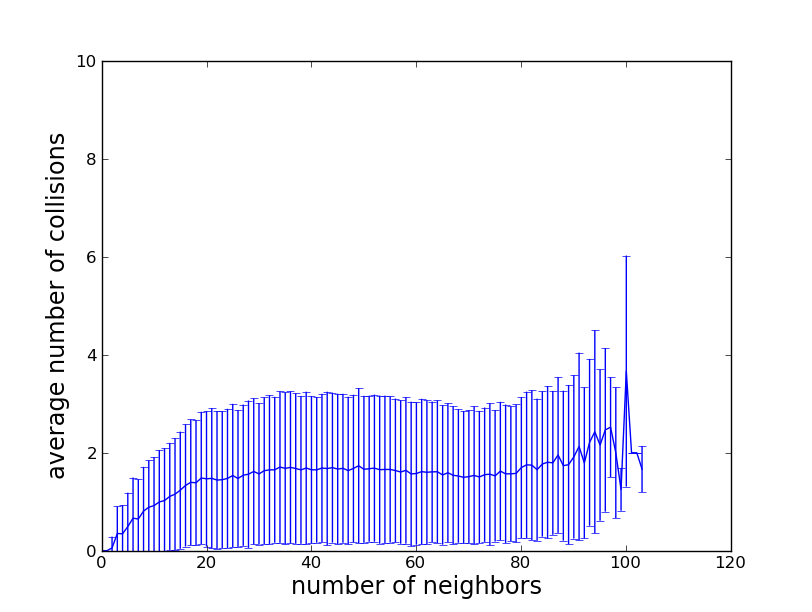}
      \label{coll_loc}}
    \hfil
    \subfigure[Log-normal shadowing model]{\includegraphics[width=2.5in, keepaspectratio=true]{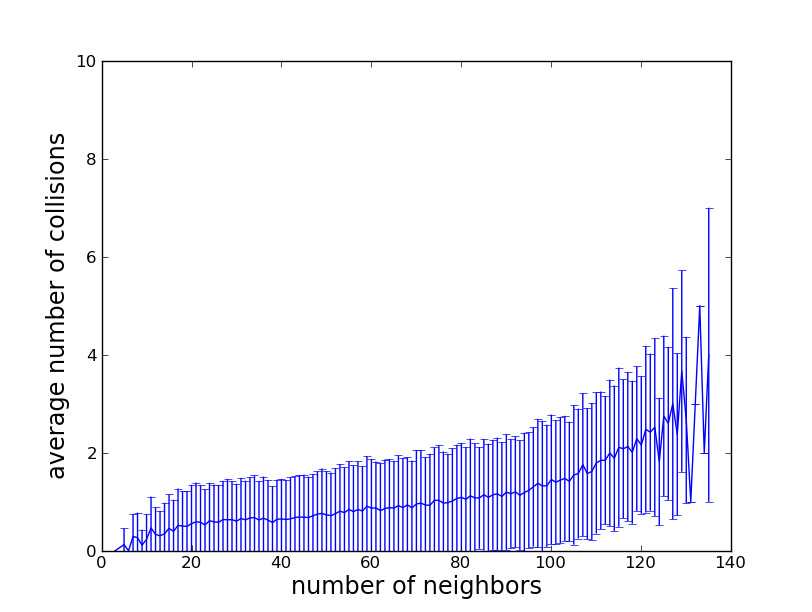}
\label{coll_loc_s}}
\caption{Average number of collision seen by a node with 95\% confidence interval in function of network density}
\end{figure}

Our goal is to monitor the number of coordinate collisions induced by method we use to construct the coordinate. Nevertheless, in the simulator the coordinate is represented by floating point numbers with finite precision which can induce collisions. Although collisions that do not come from our construction method appears, they have to be taken into account because real life implementation will also use finite precision numbers to store the coordinate. Here we study the impact of the network density on number of coordinate collisions a node sees.

Figure~\ref{coll_loc} represents the average number of collisions seen by a node for a given number of neighbors, for a given node we count the number of pairs of its neighbors having the same coordinate (i.e. the number of collisions it sees). This number is not dependent on the network density in the case of free space propagation model. It confirms the theoretical results of section 5 with the average number of collision higher than the expected number. This is due to our hypothesis in section 5 and the previously cited sources of coordinate collisions (projection, precision, etc) that we do not take into account in the theoretical analysis. From this observation we can tell that our solution better classifies nodes in dense networks because a node will see less collisions in proportion. The mean coordinate collisions number is near 2 which means that on average a node has 2 pairs of neighbor nodes that have the same coordinate. Thus for a node with 10 neighbors it is 40\% of its neighbors and for a node with 100 it is 4\%. The curves for highest densities are not very representative because there are few nodes with above 90 neighbors in our simulations, that explain the end of the curve. Figure~\ref{coll_loc_s} shows that in the case of log-normal shadowing propagation model the average of collision number seen by a node is slightly less than in the case of free space propagation model with almost the same 95\% confidence interval and it grows with the network density.

As stated previously those collisions are an issue because we want to use the coordinates to discriminate nodes in a 2-hop neighborhood. On the other hand there are quite few collisions (we see that at least 95\% of the number of collisions for any number of neighbors between 20 and 90 is below 3 with both propagation models). The solution we propose can be used on real radio chips because performances on unreliable radio links are similar to those with perfect channel.

\section{Conclusion and future works}
In this report we proposed a new VCS based on the hop-count of the nodes from the sink. Our proposition aims at differentiating the nodes in a 2 hop neighborhood while giving an information on their connectivity with the other hop-count rings. We present the theoretical background of our solution and discuss potential issues such as what we called coordinate collisions. We give an algorithm which allows the nodes to gather information they need to compute their coordinates. Simulations results show in which extend the aims of our protocol are fulfilled. We conclude that, even if we there are always coordinate collisions, the average number is very low and does not depend on the network density. Our solution is thus more suited for dense WSNs because there are less coordinates collisions in proportion. This work gives many perspectives to explore.
In the case of nodes death and birth and unreliable links the coordinates would have to be refreshed with a dynamic that depends on the births and deaths rates or/and the dynamic of the links in WSNs. We also plan to test our solution with multiple sinks. As our aim is to introduce time and reliability capabilities in WSNs mechanisms, we plan to use the characteristics of the coordinates in WSNs protocols especially at MAC and routing layers.

\bibliographystyle{plain}
\bibliography{RR}
\newpage
\tableofcontents

\end{document}